\documentclass[draft]{amsart}

\usepackage{amsmath,amssymb}
\usepackage{amsthm, amsrefs}
\usepackage{latexsym}
\usepackage{indentfirst}
\usepackage{mathrsfs}
\usepackage{graphicx}

\numberwithin{equation}{section}

\theoremstyle{plain}
\newtheorem{theorem}{Theorem}
\newtheorem*{thmmain}{Theorem $\bd{2'}$}
\newtheorem{lemma}{Lemma}[section]

\newtheorem{prop}[lemma]{Proposition}
\newtheorem{coro}[lemma]{Corollary}

\theoremstyle{definition}

\newtheorem{assump}{Assumption}

\newtheorem*{assumpA}{Assumption A$\bd{'}$}

\theoremstyle{remark}
\newtheorem*{remark}{Remark}

\DeclareMathOperator{\dist}{dist}

\DeclareMathOperator{\spec}{spec}

\newcommand{\ud}{\,\mathrm{d}}
\newcommand{\RR}{\mathbb{R}}
\newcommand{\NN}{\mathbb{N}}
\newcommand{\ZZ}{\mathbb{Z}}
\newcommand{\TT}{\mathrm{T}}

\newcommand{\bd}[1]{\boldsymbol{#1}}
\newcommand{\wt}[1]{\widetilde{#1}}
\newcommand{\wh}[1]{\widehat{#1}}

\newcommand{\mc}[1]{\mathcal{#1}}
\newcommand{\ms}[1]{\mathscr{#1}}

\newcommand{\veps}{\varepsilon}

\newcommand{\abs}[1]{\lvert#1\rvert}
\newcommand{\norm}[1]{\lVert#1\rVert}
\newcommand{\average}[1]{\langle#1\rangle}

\newcommand{\per}{\mathrm{per}}
\newcommand{\tot}{\mathrm{tot}}

\newcommand{\CB}{\mathrm{CB}}

\newcommand{\nn}{\nonumber}
\newcommand{\barint}{\kern4pt \raise3.4pt\hbox{\vrule height.6pt
    width7pt} \kern-11pt \int}

\begin{document}

\title[Cauchy-Born rule for spin-polarized TFDW model]{Cauchy-Born
  rule and spin density wave for the spin-polarized
  Thomas-Fermi-Dirac-von~Weizs\"{a}cker model}

\author{Weinan E} 
\address{Department of Mathematics and Program in Applied and
   Computational Mathematics \\
   Princeton University \\
   Princeton, NJ 08544 \\
   weinan@math.princeton.edu}

\author{Jianfeng Lu} 
\address{Department of Mathematics \\
   Courant Institute of Mathematical Sciences \\
   New York University \\ 
   New York, NY 10012 \\
   jianfeng@cims.nyu.edu}

\date{June 17, 2010}

\thanks{The work of W.~E was supported in part by the NSF grant
  DMS-0708026, grant DMS-0914336, the ONR grant N00014-01-1-0674 and
  the DOE grant DE-FG02-03ER25587. J.~Lu would like to thank Professor
  Robert V. Kohn for helpful discussions.}

\begin{abstract}
  The electronic structure (electron charges and spins) of a perfect
  crystal under external magnetic field is analyzed using the
  spin-polarized Thomas-Fermi-Dirac-von~Weizs\"{a}cker model. An
  extension of the classical Cauchy-Born rule for crystal lattices is
  established for the electronic structure under sharp stability
  conditions on charge density wave and spin density wave. A
  Landau-Lifshitz type micromagnetic energy functional is derived.
\end{abstract}

\maketitle

\section{Introduction}

The present paper is the fourth of a series of papers that are devoted
to the study of the electronic structure of smoothly deformed crystals
or crystals in an external field, by analyzing various quantum
mechanics models at different levels of complexity, including the
Kohn-Sham density functional theory, Thomas-Fermi type of models and
tight-binding models. Our overall objective is to establish the
microscopic foundation of the continuum theories of solids, such as
the nonlinear elasticity theory and the theory of magnetic materials,
in terms of quantum mechanics and to examine the boundary where the
continuum theories break down.  

In previous work \cite{ELu:ARMA, ELu:CPAM, ELu:KohnSham}, we have
studied the nonlinear tight-binding model and the Kohn-Sham density
functional theory, and established their continuum limits for smoothly
deformed crystals.  As a byproduct, we also derived macroscopic models
for the piezoelectric effect of a material from many-body quantum
theory.  In the present work, we will focus on the magnetic properties
of a material and the associated spin waves.  We choose the
spin-polarized Thomas-Fermi-Dirac-von~Weizs\"acker (TFDW) model as our
starting point.  This is a representative of a simplest version of the
density functional theory in which the electronic structure is
represented solely by the electron density instead of the wave
functions as in the Kohn-Sham theory.  As in the previous work, our
reference point is the electronic structure of a perfect crystal,
\textit{i.e.}  the equilibrium crystal lattice without external
fields. The solution of the TFDW model we are interested in is a
continuation of the solution for the unperturbed system. This works as
long as certain stability conditions are satisfied. One main objective
is to identify these stability conditions, with emphasis on the
stability of spin density waves (magnons). The overall strategy is
similar to the one in the previous paper \cite{ELu:CPAM} and
\cite{ELu:KohnSham}.

Another related work that we should mention is
\cite{BlancLeBrisLions:02} where the Thomas-Fermi-von~Weizs\"acker
model (without spin-polarization) was studied for smoothly deformed
crystals. It was shown that in the continuum limit, the total energy
of the system converges to a limiting value given by the extended
Cauchy-Born construction. There, the stability condition is
automatically satisfied since the model is convex.  \cite{ELu:07}
extended this kind of results to the tight-binding models. The
strategy in the current series of papers is quite different from those
in \cite{BlancLeBrisLions:02} or \cite{ELu:07}.  Here our emphasis is
on the stability conditions.

Throughout this paper, we use the notation $\lesssim$ for inequalities
up to an absolute constant: $f \lesssim g$ if $f \le C g$ where $C$ is
an absolute constant. Sometimes, it is more convenient to explicitly
use $C$ to denote the constant, which might change from line to line.
When it is necessary to specify the dependence of the constant on
parameters, we will use the notation $C(a, b)$ to indicate that the
constant depends on parameters $a$ and $b$. 

Standard notations are used for function spaces like $L^p$, $H^k$ and
$W^{k,p}$ so on. We also need function spaces for periodic functions.
Let $\mathbb{L} \subset \RR^3$ be a lattice with unit cell $\Gamma$.
Denote the reciprocal lattice as $\mathbb{L}^{\ast}$ and its unit cell
(the first Brillouin zone) as $\Gamma^{\ast}$. For a given $n \in
\NN$, we define
\begin{equation}
  L_n^p = \{ f\in\ms{S}'(\RR^3) \mid \tau_R f = f,\ \forall R \in
  n\mathbb{L}; \int_{n\Gamma} \abs{f}^p \ud x < \infty \}, 
\end{equation}
with norm given by
\begin{equation}
  \norm{f}_{L_n^p} = \left(n^{-3} \int_{n\Gamma} \abs{f}^p \ud x \right)^{1/p},
\end{equation} 
and similarly for $L_n^{\infty}$. Here $\tau_R$ is the translational
operator with translation vector $R$, $(\tau_R f)(x) = f(x-R)$.  The
periodic Sobolev space $H_n^k$ is defined similarly
\begin{equation}
  H_n^k = \{ f\in\ms{S}'(\RR^3) \mid \tau_R f = f,\ \forall R \in
  n\mathbb{L}; f \in H^{k}(n\Gamma) \}, \quad k \in \ZZ_+
\end{equation}
with norm
\begin{equation}
  \norm{f}_{H_n^k} = \sum_{\abs{\alpha} \leq k} \norm{\partial^{\alpha} f}_{L_n^2},
\end{equation}
where $\alpha$ denotes a multi-index and $\partial^{\alpha}$ the
corresponding partial derivative. Moreover, we also define the
periodic homogeneous Sobolev space with index $-1$ as
\begin{equation}
  \dot{H}^{-1}(n \Gamma) = \{ f \in \ms{S}'(\RR^3) \mid
  \tau_R f = f \ \forall R\in n\mathbb{L}, \sum_{k\in\mathbb{L}^{\ast}/n}
  \frac{1}{\abs{k}^2}\abs{\wh{f}(k)}^2 < \infty \}.
\end{equation}
Here, $\{\wh{f}(k)\}$ denotes the Fourier coefficients of the
$n\Gamma$-periodic function $f$
\begin{equation}
\wh{f}(k) = (2\pi)^{-3/2} \int_{n\Gamma} f(x) e^{-ik\cdot x} \ud x. 
\end{equation}
The space $\dot{H}^{-1}(n \Gamma)$ is a Hilbert space with inner
product
\begin{equation}
  \langle f, g\rangle_{\dot{H}^{-1}(n \Gamma)} = 4\pi  
  \sum_{k\in\mathbb{L}^{\ast}/n}\dfrac{1}{\abs{k}^2} 
    \overline{\wh{f}(k)}\wh{g}(k).
\end{equation}

For Banach spaces $X$ and $Y$, $\ms{L}(X, Y)$ denotes the class of bounded
linear operators from $X$ to $Y$ and $\norm{\cdot}_{\ms{L}(X,Y)}$
denotes the operator norm. When $X = Y$, we
use $\ms{L}(X) = \ms{L}(X,  X)$. 

\section{Spin-polarized Thomas-Fermi-Dirac-von~Weizs\"{a}cker model}

We consider the spin-polarized Thomas-Fermi-Dirac-von~Weizs\"{a}cker
(TFDW) model. We will restrict ourselves to the collinear
case, that is, we assume that the applied magnetic field is parallel
to a fixed axis (with possibly varying amplitude).  In the
spin-polarized TFDW model, the electronic structure is characterized
by the spin-up and spin-down densities, corresponding to the density
of spin-up and spin-down electrons. Denote  by $\rho_{+}$ and $\rho_{-}$
the spin-up and spin-down densities respectively, the energy of the
system is given by
\begin{multline}\label{eq:sTFDW}
  I^h(\rho_{+}, \rho_{-}) = \int \rho_+^{5/3} + \rho_-^{5/3} +
  \abs{\nabla \sqrt{\rho_+}}^2 + \abs{\nabla \sqrt{\rho_-}}^2
  - \int \rho_+^{4/3} + \rho_-^{4/3} \\
  + \frac{1}{2}D(\rho - \rho_b, \rho - \rho_b) - \int h m.
\end{multline}
Here $\rho = \rho_+ + \rho_-$ is the total charge density, $m = \rho_+ -
\rho_-$ is the spin density or magnetization density, and $\rho_b$ is
charge contribution from the nuclei and core 
electrons.  The shorthand notation $D(\cdot, \cdot)$ is defined as
\begin{equation*}
  D(f, g) = \iint \frac{f(x) g(y)}{\abs{x - y}} \ud x \ud y.
\end{equation*}
Finally, $h$ is the external magnetic field, which is a scalar, since we have
assumed the collinearity. 

The density satisfies the normalization constraint
\begin{equation}\label{eq:normrho}
  \int \rho = \int \rho_b,
\end{equation}
and  the positivity constraint $\rho_{+},\, \rho_{-} \geq 0$. To
deal with the positivity constraint, it is often convenient to
introduce the variables
\begin{equation*}
  \nu_+ = \sqrt{\rho_+}, \qquad \nu_- = \sqrt{\rho_-}.
\end{equation*}
In terms of $\nu_+$ and $\nu_-$, we have 
\begin{multline}\label{eq:sTFDWnu}
  I^h(\nu_{+}, \nu_{-}) = \int \nu_+^{10/3} + \nu_-^{10/3} +
  \abs{\nabla \nu_{+}}^2 + \abs{\nabla \nu_{-}}^2
  - \int \nu_+^{8/3} + \nu_-^{8/3} \\
  + \frac{1}{2} D(\rho - \rho_b, \rho - \rho_b) - \int h m,
\end{multline}
where 
\begin{equation*}
  \rho = \nu_{+}^2 + \nu_{-}^2, \qquad m = \nu_{+}^2 - \nu_{-}^2. 
\end{equation*}
The energy is given by the variational problem
\begin{equation}
  E^h = \inf_{\nu_{\pm} \in H^1(\RR^3)} I^h(\nu_+, \nu_-)
\end{equation}
with the normalization constraint \eqref{eq:normrho}. The
Euler-Lagrange equation associated with the energy functional reads,
\begin{align}
  & - \Delta \nu_+ + \tfrac{5}{3} \nu_+^{7/3} - \tfrac{4}{3}
  \nu_+^{5/3} + V \nu_+ - h \nu_+ = \lambda \nu_+; \\
  & - \Delta \nu_- + \tfrac{5}{3} \nu_-^{7/3} - \tfrac{4}{3}
  \nu_-^{5/3} + V \nu_- + h \nu_- = \lambda \nu_-; \\
  & - \Delta V = 4\pi (\rho - \rho_b),
\end{align}
where $\lambda$ is a Lagrange multiplier corresponds to the
normalization constraint.

In this work, we will consider the TFDW model for crystals with and
without the external magnetic field. We denote $\mathbb{L}$ the
underlying Bravais lattice of the crystal and $\Gamma$ the unit cell
of the crystal. The charge background $\rho_b$ is $\Gamma$-periodic
and assumed to be smooth (in other words, we are taking a
pseudo-potential approximation).  When there is no external magnetic
field, the electronic structure is characterized by the periodic-TFDW
model, with energy functional
\begin{multline}\label{eq:sTFDWunitcell}
  I_{\Gamma}^0(\nu_{+}, \nu_{-}) = \int_{\Gamma} \nu_+^{10/3} +
  \nu_-^{10/3} + \abs{\nabla \nu_{+}}^2 + \abs{\nabla \nu_{-}}^2
  - \int_{\Gamma} \nu_+^{8/3} + \nu_-^{8/3} \\
  + \frac{1}{2}D_{\Gamma}(\rho - \rho_b, \rho - \rho_b),
\end{multline}
where the periodic Coulomb interaction $D_{\Gamma}$ is given by 
\begin{equation*}
  D_{\Gamma}(f, g) = \langle f, g \rangle_{\dot{H}^{-1}(\Gamma)}.
\end{equation*}
The electronic structure is given by the variational problem 
\begin{equation}\label{eq:vpper}
  E_{\Gamma}^0 = \inf_{\nu_{\pm} \in H_1^1} I_{\Gamma}^0(\nu_{+}, \nu_{-})
\end{equation}
with the normalization constraint that
\begin{equation*}
  \int_{\Gamma} \nu_+^2 + \nu_-^2 = \int_{\Gamma} \rho_b = Z.
\end{equation*}
Here $Z$ is some fixed integer.

The Euler-Lagrange equations associated with \eqref{eq:vpper} are
given by
\begin{align}
  \label{eq:nuplusunit} & - \Delta \nu_+ + \tfrac{5}{3} \nu_+^{7/3} -
  \tfrac{4}{3}
  \nu_+^{5/3} + V \nu_+ = 0; \\
  & - \Delta \nu_- + \tfrac{5}{3} \nu_-^{7/3} - \tfrac{4}{3}
  \nu_-^{5/3} + V \nu_- = 0; \\
  \label{eq:coulombV} & - \Delta V = 4\pi (\rho - \rho_b).
\end{align}
The equations are defined on $\Gamma$ with periodic boundary
conditions. Note that we have absorbed the Lagrange multiplier into
the potential $V$, as \eqref{eq:coulombV} only determines $V$ up to a
constant. The normalization constraint is also implicitly imposed as
the solvability condition of \eqref{eq:coulombV}.

It is easy to see that the minimum of the variational problem
\eqref{eq:vpper} is achieved and the minimizers satisfy the
Euler-Lagrange equations.  The minimizers might not be unique since
the Dirac term $(-\int \nu_+^{8/3} + \nu_-^{8/3})$ is concave. Let us
take one of the minimizers, denoted as $\nu_{+, \per}$ and $\nu_{-,
  \per}$, and let us denote the corresponding potential as $V_{\per}$
(with the Lagrange multiplier included). By standard elliptic
regularity theory, it is easy to see that $\nu_{\pm, \per},\ V_{\per}
\in C^{\infty}(\Gamma)$. We extend them to the whole $\RR^3$
periodically. It is also straightforward to see that $\nu_{\pm, \per}$
are non-negative.  It is possible that $\nu_{-, \per} \equiv 0$ is a
minimizer (while the corresponding $\nu_{+, \per} > 0$).  To exclude
such cases, we will assume that there exists a positive constant
$C_{\nu}$ such that
\begin{equation}\label{eq:positive}
  \nu_{\pm, \per}(x) \geq C_{\nu}, \quad x \in \Gamma. 
\end{equation}

We may also consider the energy functional defined on the supercell $n
\Gamma$ for $n \in \NN$:
\begin{multline}\label{eq:sTFDWncell}
  I_{n \Gamma}^0(\nu_{+}, \nu_{-}) = \int_{n \Gamma} \nu_+^{10/3} +
  \nu_-^{10/3} + \abs{\nabla \nu_{+}}^2 + \abs{\nabla \nu_{-}}^2
  - \int_{n \Gamma} \nu_+^{8/3} + \nu_-^{8/3} \\
  + \frac{1}{2} D_{n \Gamma}(\rho - \rho_b, \rho - \rho_b),
\end{multline}
where 
\begin{equation*}
  D_{n\Gamma}(f, g) = \langle f, g \rangle_{\dot{H}^{-1}(n\Gamma)}.
\end{equation*}
The corresponding Euler-Lagrange equations have the same  form  as
\eqref{eq:nuplusunit}--\eqref{eq:coulombV} defined now on $n\Gamma$
with periodic boundary conditions. The functions $\nu_{\pm, \per},\,
V_{\per}$ (recall that we have extended them to $\RR^3$ periodically)
still satisfy the Euler-Lagrange equation and hence are stationary
points of the functional \eqref{eq:sTFDWncell}.

\section{Stability analysis}

\subsection{Stability of Electronic Structure}

For any $n \in \NN$, let us define the linear operator $\mc{L}$, which
is the linearization of the Euler-Lagrange equations
\eqref{eq:nuplusunit}--\eqref{eq:coulombV} on the domain $n\Gamma$, given
by
\begin{equation}
  \mc{L}   
  \begin{pmatrix}
    \omega_+ \\
    \omega_- \\
    W
  \end{pmatrix}
  = 
  \begin{pmatrix}
    \mc{L}_+ & 0 & \nu_+  \\
    0 & \mc{L}_- & \nu_- \\
    \nu_+ & \nu_- & \frac{1}{8\pi} \Delta
  \end{pmatrix}
  \begin{pmatrix}
    \omega_+ \\
    \omega_- \\
    W 
  \end{pmatrix},
\end{equation}
where the operators $\mc{L}_{\pm}$ are given by
\begin{equation}
  \mc{L}_{\pm} \omega = - \Delta \omega + \frac{35}{9} \nu_{\pm}^{4/3}
  \omega - \frac{20}{9} \nu_{\pm}^{2/3} \omega + V \omega.
\end{equation}

\begin{prop}\label{prop:boundL}
  Assume $\nu_{\pm}, \, V \in W_n^{k, \infty}$, for some $k \in \NN$. Then
  $\mc{L}$ is a bounded operator from $(H^{k+2}_n)^3$ to $(H^k_n)^3$,
  \begin{equation*}
    \norm{ \mc{L} }_{\ms{L}((H^{k+2}_n)^3, (H^k_n)^3)} \lesssim 
    \norm{\nu_+}_{W_n^{k, \infty}} + \norm{\nu_-}_{W_n^{k, \infty}} + 
    \norm{V}_{W_n^{k, \infty}} + 1.
  \end{equation*}
  Moreover, if $\nu_{\pm}, \, V \in L_n^{\infty}$, then $\mc{L}$ is a
  self-adjoint operator on the domain $\mc{D}(\mc{L}) = (H^2_n)^3$.
\end{prop}
\begin{proof}
  Denote
  \begin{equation*}
  \begin{pmatrix}
    f_+ \\ f_- \\ g
  \end{pmatrix} 
  \equiv
  \mc{L}
  \begin{pmatrix}
    \omega_+ \\ \omega_- \\ W 
  \end{pmatrix}, 
  \end{equation*}
  then we have
  \begin{align*}
    f_+ & = - \Delta \omega_+ + \frac{35}{9} \nu_{+}^{4/3} \omega_+
    - \frac{20}{9} \nu_{+}^{2/3} \omega_+ + V \omega_+ + W \nu_+; \\
    f_- & = - \Delta \omega_- + \frac{35}{9} \nu_{-}^{4/3} \omega_-
    - \frac{20}{9} \nu_{-}^{2/3} \omega_- + V \omega_- + W \nu_-; \\
    g & = \frac{1}{8\pi} \Delta W + \nu_+ \omega_+ + \nu_- \omega_-.
  \end{align*}
  It follows that
  \begin{align*}
    \norm{f_+}_{H^k_n} & \leq C(\norm{\nu_+}_{W_n^{k, \infty}},
    \norm{V}_{W_n^{k, \infty}}) \norm{\omega_+}_{H^{k+2}_n} +
    \norm{\nu_+}_{W_n^{k, \infty}} \norm{W}_{H^{k+2}_n}; \\
    \norm{f_-}_{H^k_n} & \leq C(\norm{\nu_-}_{W_n^{k, \infty}},
    \norm{V}_{W_n^{k, \infty}}) \norm{\omega_-}_{H^{k+2}_n} +
    \norm{\nu_-}_{W_n^{k, \infty}} \norm{W}_{H^{k+2}_n}; \\
    \norm{g}_{H^k_n} & \leq \frac{1}{8\pi} \norm{W}_{H^{k+2}_n} +
    \norm{\nu_+}_{W_n^{k, \infty}} \norm{\omega_+}_{H^{k+2}_n} +
    \norm{\nu_-}_{W_n^{k, \infty}} \norm{\omega_-}_{H^{k+2}_n}.
  \end{align*}
  Hence, $\mc{L}$ is bounded with the desired estimate on the operator
  norm.

  The self-adjointness of $\mc{L}$ is an easy consequence of the
  Kato-Rellich theorem \cite{ReedSimon2}, since $\nu_{\pm}$ and $V$
  viewed as multiplicative operators on $H_n^2$ are
  infinitesimally small with respect to the Laplacian operator.
\end{proof}

Consider specifically the case when  $\nu_{\pm}$ and $V$ given
by the unperturbed system:
\begin{equation}
  \mc{L}_{\per} = 
  \begin{pmatrix}
    \mc{L}_{+,\per} & 0 & \nu_{+,\per}  \\
    0 & \mc{L}_{-,\per} & \nu_{-,\per} \\
    \nu_{+,\per} & \nu_{-,\per} & \frac{1}{8\pi} \Delta
  \end{pmatrix},
\end{equation}
with $\mc{L}_{\pm, \per}$ defined by 
\begin{equation}
  \mc{L}_{\pm, \per} \omega = - \Delta \omega + \frac{35}{9} \nu_{\pm,
    \per}^{4/3} \omega - \frac{20}{9} \nu_{\pm, \per}^{2/3} \omega 
  + V_{\per} \omega.
\end{equation}

\begin{assump}[Stability of the electronic structure]\label{assump:stability}
  There exists a constant $M$ independent of $n$, such that for any
  $n$,
  \begin{equation*}
    \norm{ \mc{L}_{\per}^{-1} }_{\ms{L}((L_n^2)^3)} \leq M,
  \end{equation*}
  or equivalently, since $\mc{L}_{\per}$ is self-adjoint
  \begin{equation*}
    \dist(0, \spec(\mc{L}_{\per})) \geq 1/M.
  \end{equation*}
\end{assump}

Under the stability assumption, we can actually obtain estimates of
$\mc{L}_{\per}^{-1}$ acting on higher order Sobolev spaces. The
following result is standard, we include the proof here for
completeness.
\begin{prop}\label{prop:regularityIL}
  Under Assumption~\ref{assump:stability} and assume $\nu_{\pm}, \, V
  \in W_n^{k, \infty}$, for some $k \in \NN$.  Then we have
  \begin{equation*}
    \norm{ \mc{L}_{\per}^{-1} }_{\ms{L}((H_n^k)^3, (H_n^{k+2})^3)} \leq C(k) M.
  \end{equation*}
\end{prop}
\begin{proof}
  Let us consider $k = 0$ first. It suffices to prove the estimate
  \begin{equation}\label{eq:controlAL}
    \norm{ A \mc{L}_{\per}^{-1} u }_{(H_n^2)^3} 
    \lesssim M \norm{ u }_{(L_n^2)^3},
  \end{equation}
  for any $u = (\omega_+, \omega_-, W) \in (L_n^2)^3$, where $A$ is
  the operator
  \begin{equation*}
    A = 
    \begin{pmatrix}
      -\Delta \\
      & -\Delta \\
      & & \tfrac{1}{8\pi}\Delta
    \end{pmatrix}.
  \end{equation*}
  The left hand side of \eqref{eq:controlAL} equals to
  \begin{equation*}
    \begin{pmatrix}
      - \Delta \\
      & -\Delta \\
      & & \tfrac{1}{8\pi}\Delta 
    \end{pmatrix}
    \mc{L}_{\per}^{-1} u 
    = u - 
    \begin{pmatrix}
      F_+ & 0 & \nu_{+,\per}  \\
      0 & F_- & \nu_{-,\per} \\
      \nu_{+,\per} & \nu_{-,\per} & 0
    \end{pmatrix}
    \mc{L}_{\per}^{-1} u,
  \end{equation*}
  where 
  \begin{equation*}
    F_{\pm} = \tfrac{35}{9} \nu_{\pm, \per}^{4/3} - \tfrac{20}{9} \nu_{\pm,
      \per}^{2/3} + V_{\per}.
  \end{equation*}
  Therefore, 
  \begin{equation*}
    \begin{aligned}
      \norm{A \mc{L}_{\per}^{-1} u}_{(L_n^2)^3} & \lesssim
      \norm{u}_{(L_n^2)^3} + \max(\norm{F_{\pm}}_{L_n^{\infty}},
      \norm{\nu_{\pm, \per}}_{L_n^{\infty}})
      \norm{\mc{L}_{\per}^{-1}}_{\ms{L}((L_n^2)^3)}
      \norm{u}_{(L_n^2)^3} \\
      & \lesssim M \norm{u}_{(L_n^2)^3}.
    \end{aligned}
  \end{equation*}
  
  Suppose the statement of the Proposition is proved for $k \leq k_0$,
  let us consider $k = k_0 + 1$. Since 
  \begin{equation*}
    \begin{aligned}
      \bd{\nabla} \mc{L}_{\per}^{-1} & = \mc{L}_{\per}^{-1}
      \bd{\nabla} + [\bd{\nabla}, \mc{L}_{\per}^{-1}] \\
      & = \mc{L}_{\per}^{-1} \bd{\nabla} - \mc{L}_{\per}^{-1}
      [\bd{\nabla}, \mc{L}_{\per}] \mc{L}_{\per}^{-1},
    \end{aligned}
  \end{equation*}
  it suffices to control 
  \begin{equation*}
    \norm{ \mc{L}_{\per}^{-1} [\bd{\nabla}, \mc{L}_{\per}] 
      \mc{L}_{\per}^{-1} }_{\ms{L}((H_n^{k})^3, (H_n^{k+1})^3)},
  \end{equation*}
  where $\bd{\nabla} = I_3 \nabla$ with $I_3$ the $3 \times 3$
  identity matrix.
  Note that
  \begin{multline*}
    \norm{ \mc{L}_{\per}^{-1} [\bd{\nabla}, \mc{L}_{\per}]
      \mc{L}_{\per}^{-1} }_{\ms{L}((H_n^{k})^3, (H_n^{k+1})^3)}
    \leq \norm{ \mc{L}_{\per}^{-1} }_{\ms{L}((H_n^{k-1})^3, (H_n^{k+1})^3)} \\
    \times \norm{[\bd{\nabla}, \mc{L}_{\per}]}_{\ms{L}((H_n^{k+1})^3,
      (H_n^{k-1})^3)} \norm{\mc{L}_{\per}^{-1}
    }_{\ms{L}((H_n^{k-1})^3, (H_n^{k+1})^3)}.
  \end{multline*}
  By assumption, the Proposition holds for $k-1 = k_0$,
  hence it suffices to control the commutator $[\bd{\nabla}, \mc{L}_{\per}]$.
  An explicit calculation yields
  \begin{equation*}{}
    [\bd{\nabla}, \mc{L}_{\per}] = 
    \begin{pmatrix}
      \nabla F_+ & 0 & \nabla \nu_{+,\per}  \\
      0 & \nabla F_- & \nabla \nu_{-,\per} \\
      \nabla \nu_{+,\per} & \nabla \nu_{-,\per} & 0
    \end{pmatrix},
  \end{equation*}
  and hence the bounds follow from the regularity assumptions on
  $\nu_{\pm, \per}$ and $V_{\per}$. The proposition is proved.
\end{proof}

The Assumption~\ref{assump:stability} is stated in terms of the
operator $\mc{L}_{\per}$ acting on a series of spaces
$(L_n^2)^3$. Using the Bloch-Floquet decomposition (see
\cite{ReedSimon4} or \cite{ELu:ARMA} for an introduction), we may
obtain an equivalent characterization of the stability
assumption. Note that $\mc{L}_{\per}$ commutes with the translational
operator with respect to the lattice $\mathbb{L}$ since $\nu_{\pm,
  \per}$ and $V_{\per}$ are $\Gamma$-periodic. Denote $\Gamma^{\ast}$
as the unit cell of the reciprocal lattice of $\mathbb{L}$, the
Bloch-Floquet decomposition of $\mc{L}_{\per}$ is given by
\begin{equation}
  \mc{L}_{\per} = \barint_{\Gamma^{\ast}} \mc{L}_{\xi, \per} \ud \xi.
\end{equation}
Here for any $\xi \in \Gamma^{\ast}$, $\mc{L}_{\xi, \per}$ is the
operator 
\begin{equation}
  \mc{L}_{\xi, \per} = 
  \begin{pmatrix}
    \mc{L}_{+,\xi, \per} & 0 & \nu_{+,\per}  \\
    0 & \mc{L}_{-,\xi, \per} & \nu_{-,\per} \\
    \nu_{+,\per} & \nu_{-,\per} & \frac{1}{8\pi} \Delta_{\xi}
  \end{pmatrix},
\end{equation}
with $\mc{L}_{\pm, \xi, \per}$ given by 
\begin{equation}
  \mc{L}_{\pm, \per} \omega = - \Delta_{\xi} \omega + \frac{35}{9} \nu_{\pm,
    \per}^{4/3} \omega - \frac{20}{9} \nu_{\pm, \per}^{2/3} \omega 
  + V_{\per} \omega.
\end{equation}
We also have for any $\xi$, the operator $\mc{L}_{\xi, \per}$ defined
on the space $(L_{\xi}^2)^3$, where
\begin{equation*}
  L_{\xi}^2 = \{ f \in L_{\text{loc}}^2(\RR^3) \mid e^{i\xi\cdot x} f(x) 
  \ \Gamma\text{-periodic} \},
\end{equation*}
is self-adjoint \cite{ReedSimon4}. 

Using Bloch-Floquet decomposition of $\mc{L}_{\per}$, the stability
assumption  can also be formulated as
\begin{assumpA}
  There exists a constant $M$, such that for each $\xi \in \Gamma^{\ast}$,
  \begin{equation*}
    \norm{ \mc{L}_{\xi, \per}^{-1} }_{(L_{\xi}^2)^3} \leq M.
  \end{equation*}
\end{assumpA}

The proof of the equivalence of Assumption~\ref{assump:stability} and
Assumption~A$'$ is parallel to the corresponding results in
\cites{ELu:CPAM, ELu:KohnSham} and is standard from Bloch-Floquet
theory; hence, we omit it here.

\subsection{Example of stability and instability}

Let us consider the jellium model as an example to understand better
the stability assumptions. For the jellium model, the charge
background is a constant function $\rho_b(x) = \rho_0$. Define $\nu_0
= \tfrac{1}{2} \rho_0^{1/2}$, a solution to the Euler-Lagrange
equations are given by 
\begin{align*}
  \nu_{+}(x) = \nu_0, \quad \nu_{-}(x) = \nu_0, \quad V(x) = 0, \quad
  \lambda = \frac{5}{3}\nu_0^{4/3} - \frac{4}{3}\nu_0^{2/3}.
\end{align*}
Therefore, 
\begin{equation}
  \mc{L}_{+} \omega = \mc{L}_{-} \omega = 
  -\Delta \omega + \frac{20}{9} \nu_0^{4/3} \omega 
  - \frac{8}{9} \nu_0^{2/3} \omega.
\end{equation}

We use Fourier transform to analyze the operator $\mc{L}$.
\begin{equation}
  \mc{L}_{\xi} = 
  \begin{pmatrix}
    \xi^2 + \frac{20}{9}\nu_0^{4/3} - \frac{8}{9}\nu_0^{2/3} &
    0 & \nu_0 \\
    0 & \xi^2 + \frac{20}{9}\nu_0^{4/3} - \frac{8}{9}\nu_0^{2/3} & \nu_0 \\
    \nu_0 & \nu_0 & - \frac{1}{8\pi} \xi^2
  \end{pmatrix}
\end{equation}
One eigenvalue of the matrix is 
\begin{equation}
  \lambda_{\xi,1}  = \xi^2 + \frac{20}{9}\nu_0^{4/3} -
  \frac{8}{9}\nu_0^{2/3},
\end{equation}
corresponds to the eigenvector $(1, -1, 0)^{\TT}$. 
The other two
eigenvalues are given by
\begin{align}
  \lambda_{\xi,\pm} & = \frac{1}{2} \left( \frac{8\pi -
      1}{8\pi} \xi^2 +
    \frac{20}{9}\nu_0^{4/3} - \frac{8}{9}\nu_0^{2/3} \right. \\
  & \nn\hspace{8em} \left.  \pm \sqrt{\Bigl( \frac{8\pi + 1}{8\pi} \xi^2
      + \frac{20}{9}\nu_0^{4/3} - \frac{8}{9}\nu_0^{2/3} \Bigr)^2 + 8
      \nu_0^2} \right).
\end{align}
Observe that these two eigenvalues correspond to action of
$\mc{L}_{\xi}$ on the subspace orthogonal to the vector $(1, -1,
0)^{\TT}$, and hence, the first two components are the same.

It is easy to see that $\lambda_{\xi,1}$ becomes
positive if $\xi$ is sufficiently large. To prevent $\lambda_{\xi,1}$
from changing sign when $\xi$ is small, we need
\begin{equation*}
  \xi^2 + \frac{20}{9} \nu_0^{4/3} - \frac{8}{9} \nu_0^{2/3} > 0, \quad
  \forall \xi,
\end{equation*}
or equivalently
\begin{equation}\label{eq:condpos}
  \nu_0 > \bigl(\frac{2}{5}\bigr)^{3/2}.
\end{equation}
For the other two eigenvalues, we have
\begin{equation*}
  \lambda_{\xi,+} \lambda_{\xi,-} = 
  - \frac{1}{8\pi} \xi^2 \Bigl( \xi^2 + \frac{20}{9}\nu_0^{4/3} - \frac{8}{9} 
  \nu_0^{2/3} \Bigr) - \nu_0^2.
\end{equation*}
The product becomes negative when $\xi$ is sufficiently large, and hence
we have $\lambda_{\xi,-} < 0$ and $\lambda_{\xi,+} > 0$ for $\xi$
sufficiently large. To make sure they are nonzero for every $\xi$, we need
\begin{equation*}
  - \frac{1}{8\pi} \xi^2 \Bigl( \xi^2 + \frac{20}{9}\nu_0^{4/3} - \frac{8}{9} 
  \nu_0^{2/3} \Bigr) - \nu_0^2 < 0, \quad \forall \xi,
\end{equation*}
which is equivalent to the condition that 
\begin{equation}\label{eq:condpos2}
  \frac{20}{9}\nu_0^{4/3} - \frac{8}{9}\nu_0^{2/3} > - 8\sqrt{\pi} \nu_0. 
\end{equation}
This condition is weaker than \eqref{eq:condpos}. 

Hence, \eqref{eq:condpos} guarantees that the three eigenvalues do not
change sign for all $\xi$, and hence the matrix $\mc{L}_{\xi}$ is
non-singular. 

Let us remark that physically, the condition \eqref{eq:condpos}
corresponds to the stability of spin-density-wave, since
$\lambda_{k,1}$ corresponds to the eigenvector $(1, -1, 0)^{\TT}$,
which increases (or decreases) the spin-up component, while decreases
(or increases) the spin-down component, hence creates a
spin-density-wave. On the other hand, the condition
\eqref{eq:condpos2} corresponds to the stability of
charge-density-wave, since the spin-up and spin-down components change
together with the same amplitude. For the jellium case, since the
condition \eqref{eq:condpos2} is implied by the condition
\eqref{eq:condpos}, we observe that the spin-density-wave loses
stability earlier than the charge-density-wave when the uniform
background charge density is decreased.

\section{Cauchy-Born rule}

Let us first consider the case when the applied magnetic field is
constant $h(x) \equiv h$, and consider the cell problem
\begin{multline}
  I_{\Gamma}^h(\nu_+, \nu_-) = \int_{\Gamma} \nu_+^{10/3} +
  \nu_-^{10/3} + \abs{\nabla \nu_{+}}^2 + \abs{\nabla \nu_{-}}^2
  - \int_{\Gamma} \nu_+^{8/3} + \nu_-^{8/3} \\
  + \frac{1}{2} D_{\Gamma}(\rho - \rho_b, \rho - \rho_b) - h m_{\tot},
\end{multline}
where the periodic Coulomb interaction $D_{\Gamma}$ is given by 
\begin{equation*}
  D_{\Gamma}(f, g) = \langle f, g\rangle_{\dot{H}^{-1}(\Gamma)},
\end{equation*}
and the total magnetization $m_{\tot}$ is given by
\begin{equation*}
  m_{\tot} = \int_{\Gamma} \nu_{+}^2 - \nu_{-}^2.
\end{equation*}

Consider
\begin{equation}
  E(h) = \inf_{\nu_+, \nu_-} I^h_{\Gamma},
\end{equation}
with the normalization constraint $\int_{\Gamma} \rho = Z$. 
The Euler-Lagrange equations are given by
\begin{align}
  & - \Delta \nu_+ + \tfrac{5}{3} \nu_+^{7/3} - \tfrac{4}{3}
  \nu_+^{5/3} + (V - h) \nu_+ = 0; \\
  & - \Delta \nu_- + \tfrac{5}{3} \nu_-^{7/3} - \tfrac{4}{3}
  \nu_-^{5/3} + (V + h) \nu_- = 0; \\
  & - \Delta V = 4\pi (\rho - \rho_b);
\end{align}
in the unit cell $\Gamma$, where $h$ is a constant. Let us denote the
equations as 
\begin{equation*}
  \mc{F}(u, h) = 0,
\end{equation*}
where $u$ stands for the triple $u = (\nu_+, \nu_-, V)$. 

The following theorem shows that the solution $u$ exists provided the
stability condition is satisfied and the constant applied magnetic
field $h$ is not too large.
\begin{theorem}\label{thm:CB}
  Under Assumption~\ref{assump:stability}, there exist positive
  constants $h_0$ and $\delta$, and a unique $C^{\infty}$ map from
  $[-h_0, h_0] \to (H_1^3)^3$: $h \mapsto u(\cdot; h)$, such that
  $\norm{u(\cdot; h) - u_{\per}}_{(H_1^3)^3} \leq \delta$ and
  \begin{equation*}
    \mc{F}(u(\cdot; h), h) = 0.
  \end{equation*}
\end{theorem}

\begin{proof}
  We use the implicit function theorem. Let $u_{\per}$ be the triple
  $(\nu_{\pm, \per}, V_{\per})$, we have 
  \begin{equation*}
    \mc{F}(u_{\per}, 0) = 0.
  \end{equation*}
  Let $\delta_1$ be a positive constant to be fixed, consider the
  neighborhood around $u_{\per}$:
  \begin{equation*}
    \mc{D} = \{ u \mid \norm{ u - u_{\per} }_{(H_1^3)^3} \leq \delta_1 \}. 
  \end{equation*}
  Denote $u = (\nu_{\pm}, V)$; by Sobolev inequality, we have
  \begin{equation*}
    \norm{\nu_{\pm} - \nu_{\pm, \per}}_{L_1^{\infty}} \lesssim 
    \norm{\nu_{\pm} - \nu_{\pm, \per}}_{H_1^2}
    \leq \delta_1.
  \end{equation*}
  Take $\delta_1$ sufficiently small such that $\norm{\nu_{\pm} - \nu_{\pm,
      \per}}_{L_1^{\infty}} \leq C_{\nu} / 2$. Hence, for $u \in
  \mc{D}$, we have $\nu_{\pm} \geq C_{\nu}/2 > 0$. It is then easy to
  see that viewed as an operator from $ \mc{D} \times \RR^3 \to
  (H_1^1)^3$, $\mc{F}$ is $C^{\infty}$. Notice that
  \begin{equation*}
    \frac{\delta \mc{F}(u, h)}{\delta u} \Big\vert_{u = u_{\per}, h = 0}
    = \mc{L}_{\per} 
  \end{equation*}
  has a bounded inverse from $(H_1^1)^3$ to $(H_1^3)^3$ due to the
  stability assumption and
  Proposition~\ref{prop:regularityIL}. Applying the implicit function
  theorem on $\mc{F}$, we arrive at the desired result.

\end{proof}

\begin{remark}
  It is clear from the proof that we can consider solutions in spaces
  with higher regularity. The space $(H_1^3)^3$ is chosen here for
  proving the main result in the next section.
\end{remark}

We denote the solutions given by Theorem~\ref{thm:CB} as $\nu_{+,
  \CB}(\cdot; h)$ and $\nu_{-, \CB}(\cdot; h)$ respectively for the
spin-up and spin-down components, and $V_{\CB}(\cdot; h)$ for the
potential. Here $h$ is a parameter and $\nu_{\pm, \CB}$ and $V_{\CB}$
are $\Gamma$-periodic.

For $h \in [-h_0, h_0]$, let $\mc{L}_h$ be the linearized
operator around $u_{\CB}(\cdot, h)$:
\begin{equation}
  \mc{L}_h = \frac{\delta \mc{F}(u, h)}{\delta u}
  \Big\vert_{u = u_{\CB}(\cdot; h), h = h},
\end{equation}
given by 
\begin{equation}
  \mc{L}_h   
  \begin{pmatrix}
    \omega_+ \\
    \omega_- \\
    W
  \end{pmatrix}
  = 
  \begin{pmatrix}
    \mc{L}_{h, +} & 0 & \nu_{+, \CB}(\cdot; h)  \\
    0 & \mc{L}_{h, -} & \nu_{-, \CB}(\cdot; h)  \\
    \nu_{+, \CB}(\cdot; h) & \nu_{-, \CB}(\cdot; h) & \frac{1}{8\pi}
    \Delta
  \end{pmatrix}
  \begin{pmatrix}
    \omega_+ \\
    \omega_- \\
    W 
  \end{pmatrix},
\end{equation}
where the operators $\mc{L}_{h, +}$ and $\mc{L}_{h, -}$ are defined as
\begin{gather}
  \mc{L}_{h, +} \omega = - \Delta \omega + \frac{35}{9} \nu_{+,
    \CB}^{4/3}(x; h) \omega - \frac{20}{9} \nu_{+, \CB}^{2/3}(x; h)
  \omega + (V_{\CB}(x; h) - h)\omega; \\
  \mc{L}_{h, -} \omega = - \Delta \omega + \frac{35}{9} \nu_{-,
    \CB}^{4/3}(x; h) \omega - \frac{20}{9} \nu_{-, \CB}^{2/3}(x; h)
  \omega + (V_{\CB}(x; h) + h)\omega.
\end{gather}
By Proposition~\ref{prop:boundL} and using similar arguments as the
proof of Proposition~\ref{prop:regularityIL}, we have
\begin{prop}\label{prop:LH}
  Under the same assumptions of Theorem~\ref{thm:CB}, $\mc{L}_h$ as an
  operator from $(H_1^3)^3$ to $(H_1^1)^3$ is invertible, and the norm of
  the inverse operator is uniformly bounded for $h \in [-h_0, h_0]$. 
\end{prop}

Let us remark that Theorem~\ref{thm:CB} gives a map from $h$ to the
electronic structure ($\nu_+, \nu_-$) for the case when the external
magnetic field is homogeneous. This is slightly different from the
usual Cauchy-Born rule for crystals under deformation, where the
strain is fixed -- here $h$, the analog of stress, is fixed. One may
consider the dual problem given by
\begin{multline*}
  \wt{E}_{\CB}(m) = \inf_{\nu_{\pm} : m_{\tot} = m} \biggl\{
  \barint_{\Gamma} \nu_+^{10/3} + \nu_-^{10/3} + \abs{\nabla
    \nu_{+}}^2 + \abs{\nabla \nu_{-}}^2
  - \barint_{\Gamma} \nu_+^{8/3} + \nu_-^{8/3} \\
  + \frac{1}{2\abs{\Gamma}} D_{\Gamma}(\rho - \rho_b, \rho - \rho_b)
  \biggr\},
\end{multline*}
where $m_{\tot} = \int_{\Gamma} \nu_{+}^2 - \nu_{-}^2$ is constrained
to be equal to $m$. This can be viewed as a Legendre transform of
\begin{multline*}
  E_{\CB}(h) = \inf_{\nu_{\pm}} \barint_{\Gamma} \nu_+^{10/3} +
  \nu_-^{10/3} + \abs{\nabla \nu_{+}}^2 + \abs{\nabla \nu_{-}}^2
  - \barint_{\Gamma} \nu_+^{8/3} + \nu_-^{8/3} \\
  + \frac{1}{2\abs{\Gamma}} D_{\Gamma}(\rho - \rho_b, \rho - \rho_b) -
  \frac{1}{\abs{\Gamma}}h m_{\tot}.
\end{multline*}
The formulation in terms 
of the magnetization $m$ may bear more similarity with
the conventional Cauchy-Born rule for lattices. 

\section{Main results}

We turn to the situation when the system is under (a macroscopically
heterogeneous) external applied magnetic field. We will study the case
when the applied potential is macroscopically smooth. The ratio of the
lattice constant and the characteristic length of $h$ will serve as a
small parameter $\veps$. Given a fixed $\Gamma$-periodic function
$h(\cdot)$, two equivalent choices of scalings are possible: For any
$\veps = 1/n$ a reciprocal of positive integer, we may study a perfect
crystal with applied field $h(\veps x)$ in $n\Gamma$ with periodic
boundary condition; equivalently, we may rescale the system, so that
the lattice constant becomes $\veps$ and study the rescaled system
with applied field $h(x)$ in $\Gamma$ with periodic boundary
condition. We call the former choice the atomic unit scaling, and the
latter choice the $\veps$-scaling. We will use atomic unit scaling for
most part of the paper, however, $\veps$-scaling is more convenient
and is used for the two-scale analysis in Section~\ref{sec:twoscale}.

Under the influence of the external field,
the electronic structure of the system is determined by minimizing the
energy functional. 
\begin{multline}\label{eq:InH}
  I_n^h(\nu_{+}, \nu_{-}) = \int_{n\Gamma} \nu_+^{10/3} + \nu_-^{10/3}
  + \abs{\nabla \nu_+}^2 + \abs{\nabla \nu_-}^2
  - \int_{n\Gamma} \nu_+^{8/3} + \nu_-^{8/3} \\
  + \frac{1}{2} D_n(\rho - \rho_b, \rho - \rho_b) - \int_{n\Gamma} h(\veps x) m,
\end{multline}
where $D_n$ is the Coulomb interaction
\begin{equation*}
  D_{n}(f, g) = \langle f, g\rangle_{\dot{H}^{-1}(n \Gamma)},
\end{equation*}
and the density $\rho$ and the spin density $m$ are given by
\begin{equation*}
  \rho = \nu_+^2 + \nu_-^2, \quad m = \nu_+^2 - \nu_-^2.
\end{equation*}
The functional \eqref{eq:InH} is optimized under the normalization
constraint on the electron density 
\begin{equation}\label{eq:norm}
  n^{-3} \int_{n\Gamma} \rho = n^{-3} \int_{n\Gamma} \nu_+^2 + \nu_-^2 = Z.
\end{equation}

The Euler-Lagrange equations associated with \eqref{eq:InH} are given
by
\begin{align}
  \label{eq:nuplus} & - \Delta \nu_+ + \frac{5}{3} \nu_+^{7/3} -
  \frac{4}{3} \nu_+^{5/3} + (V - h(\veps x))\nu_+ = 0; \\
  \label{eq:numinus} & - \Delta \nu_- + \frac{5}{3} \nu_-^{7/3} -
  \frac{4}{3} \nu_-^{5/3} + (V + h(\veps x))\nu_- = 0; \\
  \label{eq:V} & - \Delta V = 4\pi \bigl(\nu_+^2 + \nu_-^2 - \rho_{b}
  \bigr),
\end{align}
in $n\Gamma$ with periodic boundary condition. Note that the
normalization constraint \eqref{eq:norm} is contained in \eqref{eq:V}
as the solvability condition. The potential $V$ is determined up to a
constant by \eqref{eq:norm}, the constant is fixed by
\eqref{eq:nuplus} and \eqref{eq:numinus} through the solvability
condition of \eqref{eq:norm}.

For later use, let us also write down the Euler-Lagrange equations in
$\veps$-scaling, which is just a rescaling of
\eqref{eq:nuplus}--\eqref{eq:V}. 
\begin{align}
  & - \veps^2 \Delta \nu_+^{\veps} + \frac{5}{3} (\nu_+^{\veps})^{7/3}
  - \frac{4}{3} (\nu_+^{\veps})^{5/3} +
  (V^{\veps} - h)\nu_+^{\veps} = 0; \\
  & - \veps^2 \Delta \nu_-^{\veps} + \frac{5}{3} (\nu_-^{\veps})^{7/3}
  - \frac{4}{3} (\nu_-^{\veps})^{5/3} +
  (V^{\veps} + h)\nu_-^{\veps} = 0; \\
  & - \veps^2 \Delta V^{\veps} = 4\pi \bigl((\nu_+^{\veps})^2 +
  (\nu_-^{\veps})^2 - \rho_{b}^{\veps} \bigr),
\end{align}
in $\Gamma$ with periodic boundary condition. We have the scaling relations
\begin{equation*}
  \nu_{\pm}^{\veps}(x) = \nu_{\pm}(x/\veps), \quad
  V^{\veps}(x) = V(x/\veps), \quad   \rho_b^{\veps}(x) =
  \rho_b(x/\veps).
\end{equation*}

In analogy with the spirit of the Cauchy-Born rule for crystal
lattices, we expect that the electronic structure around a point $x_0$
to be approximately given by the electronic structure of a crystal
under constant applied potential with amplitude $h(\veps x_0)$. As we
have discussed in the last section, the electronic structure for the
system with the constant applied potential is given by $\nu_{\pm,
  \CB}(\cdot; h(\veps x_0))$ and $V_{\CB}(\cdot; h(\veps
x_0))$. Therefore, the electronic structure constructed using the spirit
of the Cauchy-Born rule is
\begin{equation}\label{eq:CauchyBorn}
  \nu_{\pm}(x) = \nu_{\pm, \CB}(x; h(\veps x)), \quad
  V(x) = V_{\CB}(x; h(\veps x)). 
\end{equation}

One main result of this paper is that under the stability conditions,
the electronic structure constructed by the Cauchy-Born rule gives a
good approximation to a solution to the TFDW equation. In other words,
one can find a solution to the TFDW equation that is close to the
Cauchy-Born approximation.
\begin{theorem}\label{thm:main}
  Under Assumption~\ref{assump:stability}, there exist positive
  constants $h_0$, $\veps_0$ and $\delta$, such that for any $h \in
  C^{\infty}(\Gamma)$, $\norm{h}_{L^{\infty}(\Gamma)} \leq h_0$ and
  $\veps \leq \veps_0$, there exists a unique $u = (\nu_+, \nu_-, V)
  \in (H_n^2)^3$, with the properties
  \begin{itemize}
  \item $u$ is a solution to the Euler-Lagrange equation,
    \begin{equation*}
      \mc{F}(u) = 0;
    \end{equation*}
  \item $u$ is close to the approximation given by the Cauchy-Born
    rule
    \begin{equation*}
      \norm{u - u_{\CB}(x; h(\veps x))}_{(H_n^2)^3} \leq \delta \veps. 
    \end{equation*}
  \end{itemize}
\end{theorem}

We notice that once the solution to the Euler-Lagrange equation is
determined as $u = (\nu_{\pm}, V)$, the associated energy can be
written as 
\begin{multline*}
  I_n^h(\nu_{+}, \nu_{-}) = \int_{n\Gamma} \nu_+^{10/3} + \nu_-^{10/3}
  + \abs{\nabla \nu_+}^2 + \abs{\nabla \nu_-}^2
  - \int_{n\Gamma} \nu_+^{8/3} + \nu_-^{8/3} \\
  + \int_{n\Gamma} V (\rho - \rho_b) - \int_{n\Gamma} h(\veps x) m.
\end{multline*}
As a consequence of Theorem~\ref{thm:main}, the energy is well
approximated  by the Cauchy-Born rule, given by 
\begin{equation}
  I_{n, \CB}^h = \int_{n\Gamma} E_{\CB}( h(\veps x) ) \ud x
  = \frac{1}{\veps^3} \int_{\Gamma} E_{\CB}(h(x)) \ud x, 
\end{equation}
with $E_{\CB}$ defined by
\begin{multline}
  E_{\CB}(h) = \barint_{\Gamma} \nu_{+, \CB}^{10/3} + \nu_{-,
    \CB}^{10/3} + \abs{\nabla \nu_{+, \CB}}^2 + \abs{\nabla \nu_{-,
      \CB}}^2 - \barint_{n\Gamma} \nu_{+, \CB}^{8/3} + \nu_{-, \CB}^{8/3} \\
  + \barint_{\Gamma} V_{\CB} (\rho_{\CB} - \rho_b) - h
  \barint_{\Gamma} m_{\CB},
\end{multline}
where $\nu_{\pm, \CB} = \nu_{\pm, \CB}(\cdot; h)$ and similarly for
other terms. Note that, at least formally, we have
\begin{equation}\label{eq:energyfunctional}
  \int_{\Gamma} E_{\CB}(h(x)) \ud x = \inf_{m} \biggl( 
  \int_{\Gamma} \wt{E}_{\CB}(m(x)) - h(x) m(x) \ud x \biggr).
\end{equation}
Here $\wt{E}_{\CB}(m)$ is given by the cell problem
\begin{multline*}
  \wt{E}_{\CB}(m) = \inf_{\nu_{\pm} : m_{\tot} = m} \biggl\{
  \barint_{\Gamma} \nu_+^{10/3} + \nu_-^{10/3} + \abs{\nabla
    \nu_{+}}^2 + \abs{\nabla \nu_{-}}^2
  - \barint_{\Gamma} \nu_+^{8/3} + \nu_-^{8/3} \\
  + \frac{1}{2\abs{\Gamma}} D_{\Gamma}(\rho - \rho_b, \rho - \rho_b)
  \biggr\},
\end{multline*}
where $m_{\tot} = \int_{\Gamma} \nu_{+}^2 - \nu_{-}^2$ is constrained
to be equal to $m$.  

In terms of micromagnetics, the former term on the right hand side of
\eqref{eq:energyfunctional} is the anistropic term of magnetization,
and the latter term is the energy due to external magnetic field.
Compared with the usual energy functional used in micromagnetics
\cite{Brown:63, DeSimone:06}, we do not have the stray field energy
term (the nonlocal term) for the magnetostatic interaction and the
exchange term. The reason that the nonlocal term is missing is due to
the fact that the Thomas-Fermi-Dirac-von~Weizs\"{a}cker model only
contains a local term of $m$. One can try to add a term that account
for the magnetostatic interaction at the microscopic level, we would
then obtain the nonlocal term in \eqref{eq:energyfunctional}. However,
this is only natural for a non-collinear model, which will be studied
in future publications. The reason that we do not have the
exchange term in the energy functional is more fundamental. Since the
scaling we consider only covers the smooth regime, there is no hope on
the leading order to recover the exchange term which penalize change
of magnetization on the scale comparable to the atomic length
scale. One might hope to obtain the exchange term by a different
scaling limit or going to the next order, for example, zooming in the
region of domain wall. We would not go further in this direction in
the current work.

\section{Two scale analysis}

\subsection{Matched asymptotics}\label{sec:twoscale}

In this section, we use two scale analysis to build a high-order
approximate solution to the Euler-Lagrange equation.

It is more convenient to work with $\veps$-scaling in this
section. The choice of scaling is also in
agreement with standard homogenization problems. Let
us recall the Euler-Lagrange equations under $\veps$-scaling.
\begin{align*}
  & - \veps^2 \Delta \nu_+^{\veps} + \frac{5}{3} (\nu_+^{\veps})^{7/3}
  - \frac{4}{3} (\nu_+^{\veps})^{5/3} +
  (V^{\veps} - h)\nu_+^{\veps} = 0; \\
  & - \veps^2 \Delta \nu_-^{\veps} + \frac{5}{3} (\nu_-^{\veps})^{7/3}
  - \frac{4}{3} (\nu_-^{\veps})^{5/3} +
  (V^{\veps} + h)\nu_-^{\veps} = 0; \\
  & - \veps^2 \Delta V^{\veps} = 4\pi \bigl((\nu_+^{\veps})^2 +
  (\nu_-^{\veps})^2 - \rho_{b}^{\veps} \bigr).
\end{align*}
We take the following two-scale ansatz
\begin{align*}
  & \nu_+^{\veps}(x) = \nu_{+,0}(x, x/\veps) + \veps
  \nu_{+,1}(x, x/\veps) + \veps^2 \nu_{+,2}(x, x/\veps); \\
  & \nu_-^{\veps}(x) = \nu_{-,0}(x, x/\veps) + \veps
  \nu_{-,1}(x, x/\veps) + \veps^2 \nu_{-,2}(x, x/\veps); \\
  & V^{\veps}(x) = V_0(x, x/\veps) + \veps V_1(x, x/\veps) + \veps^2
  V_2(x, x/\veps).
\end{align*}
Substituting into the equations and matching orders, we obtain for the
leading order
\begin{align}
  & -\Delta_z \nu_{+, 0}(x,z) + \frac{5}{3} \nu_{+, 0}^{7/3}(x,z) -
  \frac{4}{3} \nu_{+, 0}^{5/3}(x,z) \\
  & \nn\hspace{8em} + (V_0(x,z) - h(x)) \nu_{+,0}(x,z) = 0; \\
  & -\Delta_z \nu_{-, 0}(x,z) + \frac{5}{3} \nu_{-, 0}^{7/3}(x,z) -
  \frac{4}{3} \nu_{-, 0}^{5/3}(x,z) \\
  & \nn\hspace{8em} + (V_0(x,z) + h(x)) \nu_{-,0}(x,z) = 0; \\
  & - \Delta_z V_0(x,z) = 4\pi ( \nu_{+,0}^2(x,z) + \nu_{-,0}^2(x,z) -
  \rho_b(z)).
\end{align}
The solutions are given by
\begin{equation}
  \nu_{+, 0}(x,z) = \nu_{+, \CB}(z; h(x)), \qquad 
  \nu_{-, 0}(x,z) = \nu_{-, \CB}(z; h(x)),
\end{equation}
with the potential given by 
\begin{equation}
  V_0(x,z) = V_{\CB}(z; h(x)).
\end{equation}

The next order equations are given by
\begin{align}
  & -\Delta_z \nu_{+,1} - 2\nabla_x\cdot\nabla_z \nu_{+,0} +
  \frac{35}{9} \nu_{+,0}^{4/3} \nu_{+,1} - \frac{20}{9}
  \nu_{+,0}^{2/3} \nu_{+,1} \\
  & \nn\hspace{8em} + (V_0 - h) \nu_{+,1} + V_1 \nu_{+,0} = 0; \\
  & -\Delta_z \nu_{-,1} - 2\nabla_x\cdot\nabla_z \nu_{-,0} +
  \frac{35}{9} \nu_{-,0}^{4/3} \nu_{-,1} - \frac{20}{9}
  \nu_{-,0}^{2/3} \nu_{-,1} \\
  & \nn\hspace{8em} + (V_0 + h) \nu_{-,1} + V_1 \nu_{-,0} = 0; \\
  & - \Delta_z V_1 - 2\nabla_x\cdot\nabla_z V_0 = 8\pi (\nu_{+,0}
  \nu_{+,1} + \nu_{-,0} \nu_{-,1}).
\end{align}

Using the linearized operator $\mc{L}_h$, we may rewrite the set of
equations as
\begin{equation}\label{eq:order1}
  \mc{L}_h  \begin{pmatrix}
    \nu_{+,1} \\
    \nu_{-,1} \\
    V_1
  \end{pmatrix}
  =
  \begin{pmatrix}
    f_{+,1} \\
    f_{-,1} \\
    g_1 
  \end{pmatrix},
\end{equation}
where
\begin{equation*}
  f_{+,1} = 2 \nabla_x\cdot\nabla_z \nu_{+,0}; \quad
  f_{-,1} = 2 \nabla_x\cdot\nabla_z \nu_{-,0}; \quad
  g_1 = - \tfrac{1}{4\pi} \nabla_x\cdot\nabla_z V_0.
\end{equation*}
By the regularity of $\nu_{\pm, 0}$ and $V_0$, it is easy to see that
$f_{\pm, 1}, \, g_1 \in C^{\infty}(\Gamma, H_1^1)$. Inverting
$\mc{L}_h$ by Proposition~\ref{prop:LH}, we then obtain $\nu_{\pm,1}$,
$V_1$, and 
\begin{equation*}
  \nu_{\pm, 1} \in C^{\infty}(\Gamma, H_1^3) \quad
  V_1 \in C^{\infty}(\Gamma, H_1^3).
\end{equation*}

We also have the third order equations
\begin{align}
  & -\Delta_z \nu_{+,2} - 2\nabla_x\cdot\nabla_z \nu_{+,1} - \Delta_x
  \nu_{+,0} + \frac{35}{9}\nu_{+,0}^{4/3}\nu_{+,2} + \frac{70}{27}
  \nu_{+,0}^{1/3} \nu_{+,1}^2 \\
  & \nn\qquad - \frac{20}{9} \nu_{+,0}^{2/3} \nu_{+,2} -
  \frac{20}{27} \nu_{+,0}^{-1/3} \nu_{+,1}^2 + (V_0 - h) \nu_{+,2}
  + V_1 \nu_{+,1} + V_2 \nu_{+,0} = 0; \\
  & -\Delta_z \nu_{-,2} - 2\nabla_x\cdot\nabla_z \nu_{-,1} - \Delta_x
  \nu_{-,0} + \frac{35}{9}\nu_{-,0}^{4/3}\nu_{-,2} + \frac{70}{27}
  \nu_{-,0}^{1/3} \nu_{-,1}^2 \\
  & \nn\qquad - \frac{20}{9} \nu_{-,0}^{2/3} \nu_{-,2} -
  \frac{20}{27} \nu_{-,0}^{-1/3} \nu_{-,1}^2 + (V_0 + h) \nu_{-,2}
  + V_1 \nu_{-,1} + V_2 \nu_{-,0} = 0; \\
  & - \Delta_z V_2 - 2\nabla_x\cdot\nabla_z V_1 - \Delta_x V_0 = 4\pi
  ( \nu_{+,1}^2 + 2\nu_{+, 0}\nu_{+, 2} + \nu_{-,1}^2 +
  2\nu_{-,0}\nu_{-,2} ).
\end{align}
As the first order correction, we solve
\begin{equation}\label{eq:order2}
  \mc{L}_h
  \begin{pmatrix}
    \nu_{+,2} \\
    \nu_{-,2} \\
    V_2 
  \end{pmatrix}
  =
  \begin{pmatrix}
    f_{+,2} \\
    f_{-,2} \\
    g_2 
  \end{pmatrix},
\end{equation}
where 
\begin{align*}
  & f_{+,2} = 2\nabla_x\cdot\nabla_z\nu_{+,1} + \Delta_x\nu_{+,0} -
  \frac{70}{27}\nu_{+,0}^{1/3}\nu_{+,1}^2 +
  \frac{20}{27}\nu_{+,0}^{-1/3}\nu_{+,1}^2 - V_1 \nu_{+,1}; \\
  & f_{-,2} = 2\nabla_x\cdot\nabla_z\nu_{-,1} + \Delta_x\nu_{-,0} -
  \frac{70}{27}\nu_{-,0}^{1/3}\nu_{-,1}^2 +
  \frac{20}{27}\nu_{-,0}^{-1/3}\nu_{-,1}^2 - V_1 \nu_{-,1}; \\
  & g_2 = -\frac{1}{8\pi}(2\nabla_x\cdot\nabla_z V_1 + \Delta_x V_0) -
  \frac{1}{2}(\nu_{+,1}^2 + \nu_{-,1}^2).
\end{align*}
Therefore, \eqref{eq:order2} is solvable to give 
$\nu_{\pm,2}$ and $V_2$ in $C^{\infty}(\Gamma, H_1^3)$.

This procedure can be carried on for even higher order terms. We
omit the details here.

\begin{remark}
  There is an important difference between the Thomas-Fermi type of
  models and the Kohn-Sham type of models considered in
  \cite{ELu:KohnSham}. In the two scale analysis for the Kohn-Sham map
  developed in \cite{ELu:KohnSham}, the macroscopic part of the
  potential on the leading order depends on the density on the order
  of $\veps^2$, making the closure a bit unusual. Here, the
  macroscopic part $\average{V_0}$ is determined on the leading order,
  and it imposes a constraint on the third order densities.  In
  particular, as observed in \cite{ELu:07} and \cite{GarciaLuE:07}, the Coulomb
  potential in Thomas-Fermi type of models are determined locally to the
  leading order, while it is not the case for Kohn-Sham type
  of models. This also leads to important differences in developing
  multiscale algorithms for these two type of models.
\end{remark}

\subsection{Approximate solution}

Let us take the approximate solution built in the last section:
\begin{align*}
  & \nu_+(x) = \nu_{+,0}(\veps x, x) + \veps
  \nu_{+,1}(\veps x, x) + \veps^2 \nu_{+,2}(\veps x, x); \\
  & \nu_-(x) = \nu_{-,0}(\veps x, x) + \veps
  \nu_{-,1}(\veps x, x) + \veps^2 \nu_{-,2}(\veps x, x); \\
  & V(x) = V_0(\veps x, x) + \veps V_1(\veps x, x) + \veps^2 V_2(\veps
  x, x).
\end{align*}
Here we have rescaled the functions into the units in which the lattice
parameter is $1$.

\begin{prop}\label{prop:approximate}
  \begin{equation*}
    \norm{\mc{F}(\nu_{+}, \nu_{-}, 
      V)}_{(L_n^2)^3} \lesssim \veps^3. 
  \end{equation*}
\end{prop}

\begin{proof}
  Denote
  \begin{equation*}
    (f_{+}, f_{-}, g) = \mc{F}(\nu_{+}, \nu_{-}, V).
  \end{equation*}
  We write
  \begin{equation}
    \begin{aligned}
      f_{+} & = - \Delta \nu_+^{\veps} + \frac{5}{3}
      (\nu_+^{\veps})^{7/3} - \frac{4}{3} (\nu_+^{\veps})^{5/3} +
      (V^{\veps} - h)\nu_+^{\veps} \\
      & = f_{+,1} + f_{+,2} + f_{+,3} + f_{+, 4},
    \end{aligned}
  \end{equation}
  where we have introduced the shorthand notation
 \begin{align*}
    & f_{+,1} = - \veps^3 \Delta_x \nu_{+,1} - 2\veps^3 \nabla_x \cdot
    \nabla_z \nu_{+,2} - \veps^4 \Delta_x \nu_{+,2}; \\
    & f_{+,2} = \frac{5}{3} ( \nu_{+,0} + \veps \nu_{+,1} +
    \veps^2 \nu_{+,2} )^{7/3} - \frac{5}{3} \nu_{+,0}^{7/3} \\
    & \hspace{8em} - \frac{35}{9} \veps \nu_{+,0}^{4/3} (\nu_{+,1} +
    \veps \nu_{+,2}) - \frac{70}{27} \veps^2
    \nu_{+,0}^{1/3} \nu_{+,1}^2; \\
    & f_{+,3} = - \frac{4}{3} ( \nu_{+,0} + \veps \nu_{+,1} + \veps^2
    \nu_{+,2} )^{5/3} + \frac{4}{3} \nu_{+,0}^{5/3} \\
    & \hspace{8em} + \frac{20}{9} \veps \nu_{+,0}^{2/3} (\nu_{+,1} +
    \veps \nu_{+,2}) + \frac{20}{27} \veps^2 \nu_{+,0}^{-1/3} \nu_{+,1}^2; \\
    & f_{+,4} = \veps^3 V_1 \nu_{+,2} + \veps^3 V_2 \nu_{+,1} +
    \veps^4 V_2 \nu_{+,2}.
  \end{align*}
  Here $\nu$ and $V$ are evaluated at $(\veps x, x)$. Since $
  \norm{f}_{L_n^2} \leq \norm{f}_{L_n^{\infty}}$, it suffices to prove
  that $\norm{f_{+, i}}_{L_n^{\infty}} \lesssim \veps^3 $ for $i = 1,
  2, 3, 4$.  Since we have $\nu_{+, 1},\, \nu_{+, 2} \in
  C^{\infty}(\Gamma, H_1^3)$, by Sobolev embedding $H_1^2 \subset
  L^{\infty}(\Gamma)$, it is easy to see that
  $\norm{f_{+,1}}_{L_n^{\infty}} \lesssim \veps^3$. The estimates for
  $f_{+,2}$ and $f_{+,3}$ follow from Taylor expansion. Finally,
  the desired estimate of $f_{+,4}$ is obtained from Sobolev embedding
  applied on $V_1$, $V_2$, $\nu_{+,1}$ and $\nu_{+,2}$. In summary, we have
  \begin{equation*}
    \norm{f_+}_{L_n^2} \lesssim \veps^3.
  \end{equation*}

  The argument for $f_{-}$ is completely the same as that for
  $f_{+}$. Let us consider
  \begin{equation}
    \begin{aligned}
      g & = - \veps^2 \Delta V^{\veps} - 4\pi \bigl((\nu_+^{\veps})^2
      + (\nu_-^{\veps})^2 - \rho_{b}^{\veps} \bigr) \\
      & = g_1 + g_2 + g_3,
    \end{aligned}
  \end{equation}
  where we have introduced the shorthands
  \begin{align*}
    g_1 & = - \veps^3 \Delta_x V_1 - 2\veps^3 \nabla_x\cdot\nabla_z V_2 -
    \veps^4 \Delta_x V_2; \\
    g_2 & = 8\pi \veps^3 \nu_{+,1}\nu_{+,2} + 4\pi \veps^4 \nu_{+,2}^2; \\
    g_3 & = 8\pi \veps^3 \nu_{-,1}\nu_{-,2} + 4\pi \veps^4 \nu_{-,2}^2. 
  \end{align*}
  It is clear that by analogous argument as above, we have the
  estimate $\norm{g}_{L_n^2} \lesssim \veps^3$.

\end{proof}

\section{Proof of Theorem~\ref{thm:main}}
Similar to \cites{ELu:CPAM, ELu:KohnSham}, we use Newton-Raphson
iteration to find a solution to the Euler-Lagrange equation in the
neighborhood of the approximated solution constructed above.

We will start with the approximate solution we constructed 
\begin{equation}\label{eq:defu0}
  u^0 = (\nu_+^0, \nu_-^0, V^0) \in \mc{D}_n,
\end{equation}
where
\begin{align*}
  & \nu_+^0(x) = \nu_{+,0}(\veps x, x) + \veps
  \nu_{+,1}(\veps x, x) + \veps^2 \nu_{+,2}(\veps x, x); \\
  & \nu_-^0(x) = \nu_{-,0}(\veps x, x) + \veps
  \nu_{-,1}(\veps x, x) + \veps^2 \nu_{-,2}(\veps x, x); \\
  & V^0(x) = V_0(\veps x, x) + \veps V_1(\veps x, x) + \veps^2
  V_2(\veps x, x).
\end{align*}
Here the superscript $0$ is used to indicate the initial point 
for the Newton iteration. Note that we have rescaled the
functions, so that $\nu_{\pm}^0$ and $V^0$ are defined in $n\Gamma$.

We need several additional lemmas for the proof of
Theorem~\ref{thm:main}. The following lemmas are proved under the
assumptions in the statement of
Theorem~\ref{thm:main}.

\begin{lemma}\label{lem:compare1}
  For any $\kappa > 0$, there exists positive constants $h_0$ and
  $\veps_0$, such that for all $\veps \leq \veps_0$ and $h$ with
  $\norm{h}_{L^{\infty}(\Gamma)} \leq h_0$, we have
  \begin{equation*}
    \norm{\mc{L}_{u^0} - \mc{L}_{\per}}_{\ms{L}((L_n^2)^3)} \leq 
    \kappa.
  \end{equation*}
\end{lemma}
\begin{proof}
  By definition, given $(\omega_{+}, \omega_{-}, W) \in
  (H_n^2)^3$, we have
  \begin{equation}
    ( \mc{L}_{u^0} - \mc{L}_{\per} ) 
    \begin{pmatrix}
      \omega_{+} \\
      \omega_{-} \\
      W 
    \end{pmatrix}
    = 
    \begin{pmatrix}
      \delta\mc{L}_{+} & 0 & \delta\nu_+  \\
      0 & \delta\mc{L}_{-} & \delta\nu_-  \\
      \delta\nu_+ & \delta\nu_- & 0  
    \end{pmatrix}
    \begin{pmatrix}
      \omega_{+} \\
      \omega_{-} \\
      W 
    \end{pmatrix}
    \equiv
    \begin{pmatrix}
      f_{+} \\ f_{-} \\ g 
    \end{pmatrix},
  \end{equation}
  where the last equality serves as a definition and we have
  introduced the notations
  \begin{equation*}
    \delta \mc{L}_{+} = 
    \frac{35}{9} \Bigl( (\nu_{+}^0)^{4/3} - \nu_{+, \per}^{4/3}
    \Bigr) - \frac{20}{9} \Bigl( (\nu_{+}^0)^{2/3} - \nu_{+, \per}^{2/3} \Bigr)
    + (V^0 - h - V_{\per}),
  \end{equation*}
  \begin{equation*}
    \delta \mc{L}_{-} = 
    \frac{35}{9} \Bigl( (\nu_{-}^0)^{4/3} - \nu_{-, \per}^{4/3}
    \Bigr) - \frac{20}{9} \Bigl( (\nu_{-}^0)^{2/3} - \nu_{-, \per}^{2/3} \Bigr)
    + (V^0 + h - V_{\per}),
  \end{equation*}
  and 
  \begin{equation*}
    \delta\nu_{\pm} = \nu^0_{\pm} - \nu_{\pm, \per}. 
  \end{equation*}
  
  Note that
  \begin{equation*}
    \begin{aligned}
      \delta\nu_{+}(x) & = \nu_{+,0}(\veps x, x) + \veps
      \nu_{+,1}(\veps x, x) + \veps^2 \nu_{+,2}(\veps x, x) - \nu_{+,
        \per}(x) \\
      & = \bigl(\nu_{+, \CB}(x; h(\veps x)) - \nu_{+, \CB}(x; 0)\bigr)
      + \veps \nu_{+,1}(\veps x, x) + \veps^2 \nu_{+,2}(\veps x, x).
    \end{aligned}
  \end{equation*}
  By the smooth dependence on $\nu_{+, \CB}(\cdot; h)$ on the
  parameter $h$, we have
  \begin{equation*}
    \norm{\nu_{+, \CB}(x; h(\veps x)) - \nu_{+, \CB}(x; 0)}_{L^{\infty}(n\Gamma)}
    \lesssim \norm{h}_{L^{\infty}(\Gamma)}.
  \end{equation*}
  The other two terms are of higher order in $\veps$, since  $\nu_{+,1}(\veps
  x, x)$ and $\nu_{+,2}(\veps x, x)$ are bounded uniformly in
  $n\Gamma$. Therefore, we obtain
  \begin{equation*}
    \norm{\delta\nu_{+}}_{L^{\infty}(n\Gamma)} \lesssim \norm{h}_{L^{\infty}(\Gamma)} 
    + \veps.
  \end{equation*}
  Obviously, the same estimates also hold for $\delta\nu_{-}$. Hence,
  \begin{equation*}
    \begin{aligned}
      \norm{g}_{L_n^2} & \leq \norm{\delta
        \nu_+}_{L^{\infty}(n\Gamma)} \norm{\omega_+}_{L_n^2} +
      \norm{\delta \nu_-}_{L^{\infty}(n\Gamma)} \norm{\omega_-}_{L_n^2} \\
      & \lesssim (\norm{h}_{L^{\infty}(\Gamma)} + \veps)
      (\norm{\omega_+}_{L_n^2} + \norm{\omega_-}_{L_n^2}).
    \end{aligned}
  \end{equation*}
    
  The analysis for $f_+$ and $f_-$ are the same, let us study $f_+$.
  Using Taylor expansion, we have
  \begin{equation*}
    \Bigl\lvert (\nu_+^0)^{4/3}(x) - (\nu_{+,\per})^{4/3}(x) \Bigr\rvert
    \leq \frac{4}{3}  \Bigl\lvert \nu_+^0(x) - \nu_{+,\per}(x) \Bigr\rvert 
    \max_{\nu \in [\nu_{+,\mathrm{min}}(x), \nu_{+, \mathrm{max}}(x)]} \nu^{1/3},
  \end{equation*}
  where $\nu_{+,\mathrm{min}}(x) = \min(\nu_+^0(x), \nu_{+,\per}(x))$
  and $\nu_{+,\mathrm{max}}(x) = \max(\nu_+^0(x),
  \nu_{+,\per}(x))$. Since
  \begin{equation*}
    \abs{\nu_+^0(x) - \nu_{+, \per}(x)} \leq 
    \norm{\delta \nu_+^0 }_{L^{\infty}} \lesssim 
    \norm{h}_{L^{\infty}(\Gamma)} + \veps,
  \end{equation*}
  we have for $h_0$ and $\veps_0$ sufficiently small, $\nu_+^0(x)$ is
  bounded from above and also from below away from zero uniformly for
  $x \in n\Gamma$. Hence $\max_{\nu \in [\nu_{+,\mathrm{min}}(x),
    \nu_{+, \mathrm{max}}(x)]} \nu^{1/3}$ is bounded. Therefore,
  \begin{equation*}
    \bigl\lvert (\nu_+^0)^{4/3}(x) - (\nu_{+,\per})^{4/3}(x) \bigr\rvert
    \lesssim \norm{h}_{L^{\infty}(\Gamma)} + \veps.
  \end{equation*}
  It follows that
  \begin{equation}\label{eq:f1}
    \biggl\lVert \frac{35}{9} \bigl((\nu_+^0)^{4/3} - (\nu_{+,\per})^{4/3}\bigr)
    \omega_+ \biggr\rVert_{L_n^2} \lesssim \bigl(\norm{h}_{L^{\infty}(\Gamma)} 
    + \veps \bigr) \norm{\omega_+}_{L_n^2}.
  \end{equation}
  Using similar arguments, we have
  \begin{equation}\label{eq:f2}
    \biggl\lVert \frac{20}{9} \bigl((\nu_+^0)^{2/3} - (\nu_{+,\per})^{2/3}\bigr)
    \omega_+ \biggr\rVert_{L_n^2} \lesssim \bigl(\norm{h}_{L^{\infty}(\Gamma)} 
    + \veps \bigr) \norm{\omega_+}_{L_n^2}.
  \end{equation}
  Compare the difference of $V^0$ and $V_{\per}$, we have
  \begin{equation*}
    V^0 - V_{\per} = V_{\CB}(x; h(\veps x)) - V_{\CB}(x; 0) 
    + \veps V_1(\veps x, x) + \veps^2 V_2(\veps x, x). 
  \end{equation*}
  Analogous to the control of $\delta \nu_+$, we have
  \begin{equation*}
    \norm{V^0 - V_{\per}}_{L^{\infty}(n\Gamma)} \lesssim 
    \norm{h}_{L^{\infty}(\Gamma)} + \veps.
  \end{equation*}
  Therefore, 
  \begin{equation}\label{eq:f3}
    \begin{aligned}
      \norm{(V^0 - H - V_{\per}) \omega_+}_{L_n^2} & \leq \norm{V^0 -
        h - V_{\per}}_{L^{\infty}(n\Gamma)} \norm{\omega_+}_{L_n^2} \\
      & \lesssim (\norm{h}_{L^{\infty}(\Gamma)} + \veps)
      \norm{\omega_+}_{L_n^2}.
    \end{aligned}
  \end{equation}
  We also have
  \begin{equation}\label{eq:f4}
    \norm{\delta \nu_+ W}_{L_n^2} \leq \norm{\delta \nu_+}_{L^{\infty}(n\Gamma)}
    \norm{W}_{L_n^2} \lesssim (\norm{h}_{L^{\infty}(\Gamma)} + \veps)
    \norm{W}_{L_n^2}.
  \end{equation}
  Combining \eqref{eq:f1}--\eqref{eq:f4}, we obtain
  \begin{equation*}
    \norm{f_+}_{L_n^2} \lesssim (\norm{h}_{L^{\infty}(\Gamma)} + \veps)
    (\norm{\omega_+}_{L_n^2} + \norm{W}_{L_n^2}).
  \end{equation*}
  Therefore,
  \begin{equation*}
    \norm{\mc{L}_{u^0} - \mc{L}_{\per}}_{\ms{L}((L_n^2)^3)} \lesssim
    \norm{h}_{L^{\infty}(\Gamma)} + \veps.
  \end{equation*}
  The conclusion of the Lemma follows.
\end{proof}

\begin{coro}\label{coro:invL0}
  There exists positive constants $h_0$ and
  $\veps_0$, such that for all $\veps \leq \veps_0$ and $h$ with
  $\norm{h}_{L^{\infty}(\Gamma)} \leq h_0$, we have
  \begin{equation*}
    \norm{\mc{L}_{u^0}^{-1}}_{\ms{L}((L_n^2)^3, (H_n^2)^3)} \lesssim 1.
  \end{equation*}
  In particular, the bound is independent of $n$.
\end{coro}
\begin{proof}
  Since
  \begin{equation*}
    \mc{L}_{u^0} = \mc{L}_{\per} + (\mc{L}_{u^0} - \mc{L}_{\per}),
  \end{equation*}
  we have 
  \begin{equation*}
    \mc{L}_{u^0}^{-1} = \mc{L}_{\per}^{-1} \bigl( 
    \mc{I} + (\mc{L}_{u^0} - \mc{L}_{\per})\mc{L}_{\per}^{-1} \bigr)^{-1}
  \end{equation*}
  if the right hand side is well-defined, where $\mc{I}$ is the
  identity operator.

  By Assumption~\ref{assump:stability}, there exists $\kappa > 0$ such
  that $ \norm{\mc{L}_{u^0} - \mc{L}_{\per}}_{\ms{L}((H_n^2)^3,
    (L_n^2)^3)} \leq \kappa$ implies
  \begin{equation*}
    \norm{(\mc{L}_{u^0} - \mc{L}_{\per})\mc{L}_{\per}^{-1}}_{\ms{L}((H_n^2)^3)}
    \leq 1/2,
  \end{equation*}
  and hence $\mc{I} + (\mc{L}_{u^0} -
  \mc{L}_{\per})\mc{L}_{\per}^{-1}$ is invertible on $(H_n^2)^3$ with
  the norm of the inverse less than $2$. Therefore, by
  Lemma~\ref{lem:compare1}, there exists $h_0$ and $\veps_0$
  sufficiently small that
  \begin{equation*}
    \norm{\mc{L}_{u^0}^{-1}}_{\ms{L}((L_n^2)^3, (H_n^2)^3)}
    \leq 2 \norm{\mc{L}_{\per}^{-1}}_{\ms{L}((L_n^2)^3, (H_n^2)^3)}.
  \end{equation*}
  The corollary is proved by combining the above inequality with
  Assumption~\ref{assump:stability}. 
\end{proof}

\begin{lemma}\label{lem:compare2}
  If $u, u'$ satisfy $\norm{u - u^0}_{(H_n^2)^3}\leq \gamma \veps^3$,
  $\norm{u' - u^0}_{(H_n^2)^3}\leq \gamma \veps^3$,  then
  \begin{equation*}
    \norm{ \mc{L}_u - \mc{L}_{u'} }_{\ms{L}((L_n^2)^3)}
    \lesssim C(\gamma) \veps^{-3/2} \norm{ u - u' }_{(H_n^2)^3}.
  \end{equation*}
\end{lemma}

\begin{proof}
  By definition, given $(\omega_{+}, \omega_{-}, W) \in (L_n^2)^3$, we
  have
  \begin{equation}
    ( \mc{L}_{u} - \mc{L}_{u'} ) 
    \begin{pmatrix}
      \omega_{+} \\
      \omega_{-} \\
      W 
    \end{pmatrix}
    = 
    \begin{pmatrix}
      \delta\mc{L}_{+} & 0 & \delta\nu_+ \\
      0 & \delta\mc{L}_{-} & \delta\nu_- \\
      \delta\nu_+ & \delta\nu_- & 0  
    \end{pmatrix}
    \begin{pmatrix}
      \omega_{+} \\
      \omega_{-} \\
      W 
    \end{pmatrix}
    \equiv
    \begin{pmatrix}
      f_{+} \\ f_{-} \\ g
    \end{pmatrix},
  \end{equation}
  where the last equality serves as a definition and we have
  introduced the notations
  \begin{equation*}
    \delta \mc{L}_{\pm} = 
    \frac{35}{9} \Bigl( (\nu_{\pm})^{4/3} - (\nu_{\pm}')^{4/3}
    \Bigr) - \frac{20}{9} \Bigl( (\nu_{\pm})^{2/3} - (\nu_{\pm}')^{2/3} \Bigr)
    + (V - V'),
  \end{equation*}
  and 
  \begin{equation*}
    \delta\nu_{\pm} = \nu_{\pm} - \nu_{\pm}'. 
  \end{equation*}

  First let us control $g$. By Sobolev inequality, we have
  \begin{equation*}
    \norm{\delta\nu_{\pm}}_{L^{\infty}(n \Gamma)} 
    \lesssim \norm{\delta\nu_{\pm}}_{H^{2}(n \Gamma)} 
    = n^{3/2} \norm{\delta\nu_{\pm}}_{H^2_n}. 
  \end{equation*}
  Hence, 
  \begin{equation*}
    \begin{aligned}
      \norm{g}_{L^2_n} & = \norm{\delta\nu_+ \omega_+ +
        \delta\nu_- \omega_-}_{L^2_n} \\
      & \leq \norm{\delta\nu_+}_{L^{\infty}(n\Gamma)}
      \norm{\omega_+}_{L_n^2} +
      \norm{\delta\nu_-}_{L^{\infty}(n\Gamma)}
      \norm{\omega_-}_{L_n^2} \\
      & \leq n^{3/2} \Bigl(
      \norm{\delta\nu_{+}}_{H^2_n}\norm{\omega_+}_{L_n^2} +
      \norm{\delta\nu_{-}}_{H^2_n}\norm{\omega_-}_{L_n^2} \Bigr).
    \end{aligned}
  \end{equation*}
  Similarly, we have
  \begin{equation*}
    \norm{ \delta\nu_+ W }_{L^2_n} 
    \leq n^{3/2} \norm{\delta\nu_+}_{H^2_n} \norm{W}_{L^2_n}.
  \end{equation*}
  Now consider the term $\delta\mc{L}_+\omega_+$, we have
  \begin{equation*}
    \bigl\lvert (\nu_+)^{4/3}(x) - (\nu_+')^{4/3}(x) \bigr\rvert
    \leq \frac{4}{3}  \abs{\nu_+(x) - \nu_+'(x)} 
    \max_{\nu \in [\nu_{+,\mathrm{min}}(x), \nu_{+, \mathrm{max}}(x)]} \nu^{1/3},
  \end{equation*}
  where $\nu_{+,\mathrm{min}}(x) = \min(\nu_+(x), \nu_+'(x))$ and
  $\nu_{+,\mathrm{max}}(x) = \max(\nu_+(x), \nu_+'(x))$. Since
 \begin{equation*}
   \abs{\nu_+(x) - \nu_+^0(x)} \leq 
   \norm{\nu_+ - \nu_+^0}_{L^{\infty}} \lesssim 
    n^{3/2} \norm{\nu_+ - \nu_+^0(x)}_{H^2_n} \leq \gamma \veps^{3/2},
  \end{equation*}
  we have for $\veps$ sufficiently small, $\nu_+(x)$ is bounded from
  above and also from below away from zero uniformly for $x \in
  n\Gamma$. The same holds for $\nu_+'(x)$, and hence $\max_{\nu \in
    [\nu_{+,\mathrm{min}}(x), \nu_{+, \mathrm{max}}(x)]} \nu^{1/3}$ is
  bounded. Therefore, 
  \begin{equation*}
    \bigl\lvert (\nu_+)^{4/3}(x) - (\nu_+')^{4/3}(x) \bigr\rvert
    \lesssim C(\gamma) \abs{\nu_+(x) - \nu_+'(x)} \lesssim
    C(\gamma) n^{3/2} \norm{\nu_+ - \nu_+'}_{H^2_n}.
  \end{equation*}
  It follows that
  \begin{equation*}
    \biggl\lVert \frac{35}{9} \bigl((\nu_+)^{4/3} - (\nu_+')^{4/3}\bigr)
    \omega_+ \biggr\rVert_{L_n^2} \lesssim C(\gamma) 
    n^{3/2} \norm{\nu_+ - \nu_+'}_{H^2_n}
    \norm{\omega_+}_{L_n^2}.
  \end{equation*}
  The argument for the remaining two terms in $\delta\mc{L}_+
  \omega_+$ is similar, and we conclude that
  \begin{equation*}
    \norm{\delta\mc{L}_+ \omega_+}_{L_n^2} \lesssim C(\gamma) 
    n^{3/2} \norm{u - u'}_{(H^2_n)^3}
    \norm{\omega_+}_{L_n^2},
  \end{equation*}
  and analogously
  \begin{equation*}
    \norm{\delta\mc{L}_- \omega_-}_{L_n^2} \lesssim C(\gamma)
    n^{3/2} \norm{u - u'}_{(H^2_n)^3}
    \norm{\omega_-}_{L_n^2}.
  \end{equation*}
  In summary, we have $\norm{\delta\mc{L}}_{\ms{L}((L_n^2)^3)} \lesssim
  C(\gamma) n^{3/2} \norm{u - u'}_{(H_n^2)^3} = C(\gamma) \veps^{-3/2}
  \norm{u - u'}_{(H_n^2)^3}$.
\end{proof}

\begin{lemma}\label{lem:diffL}
  If $u, u'$ satisfy $\norm{u - u^0}_{(H_n^2)^3}\leq \gamma \veps^3$,
  $\norm{u' - u^0}_{(H_n^2)^3}\leq \gamma \veps^3$,  then
  \begin{equation*}
    \norm{ \mc{L}_{u^0}^{-1} \bigl(\mc{F}(u) 
      - \mc{F}(u') - \mc{L}_{u^0}(u - u')\bigr) }_{(H_n^2)^3} 
    \lesssim C(\gamma) \veps^{3/2} \norm{ u - u' }_{(H_n^2)^3}.
  \end{equation*}
\end{lemma}
\begin{proof}
  We write
  \begin{equation*}
    \mc{F}(u) - \mc{F}(u') - \mc{L}_{u^0}(u-u') = 
    \int_0^1 (\mc{L}_{u_t} - \mc{L}_{u^0} ) (u - u') \ud t,
  \end{equation*}
  where $u_t = t u + (1-t) u'$. It is easy to see that for any $t\in
  [0,1]$, we have
  \begin{equation*}
    \norm{u_t - u^0}_{(H_n^2)^3} \leq \max\Bigl(\norm{u - u^0}_{(H_n^2)^3},
    \norm{u' - u^0}_{(H_n^2)^3}\Bigr).
  \end{equation*}
  Therefore, using Lemma~\ref{lem:compare2}, we obtain
  \begin{equation*}
    \begin{aligned}
      \norm{\mc{F}(u) - \mc{F}(u') - \mc{L}_{u^0}(u-u')}_{(L_n^2)^3} &
      \leq \sup_{t \in[0, 1]} \norm{ \mc{L}_{u_t} - \mc{L}_{u^0}
      }_{\ms{L}((L_n^2)^3)} \norm{u-u'}_{(L_n^2)^3} \\
      & \lesssim \veps^{-3/2} C(\gamma) \sup_{t\in[0,1]} \norm{u_t -
        u^0}_{(H_n^2)^3} \norm{u-u'}_{(L_n^2)^3} \\
      & \leq C(\gamma) \veps^{3/2} \norm{u-u'}_{(L_n^2)^3}.
    \end{aligned}
  \end{equation*}
  We conclude using Corollary~\ref{coro:invL0}.
\end{proof}

Now we are ready to prove the main result. We will actually prove a
stronger version of Theorem~\ref{thm:main} with higher order error
estimate. Theorem~\ref{thm:main} is clearly a corollary of the
following result.
\begin{thmmain}
  Under the same assumptions as Theorem~\ref{thm:main}, there exists a
  unique $u = (\nu_+, \nu_-, V) \in (H_n^2)^3$ such that
  \begin{itemize}
  \item $u$ is a solution to the Euler-Lagrange equation,
    \begin{equation*}
      \mc{F}(u) = 0;
    \end{equation*}
  \item $u$ is close to $u^0$
    \begin{equation*}
      \norm{u - u^0}_{(H_n^2)^3} \leq \delta \veps^3. 
    \end{equation*}
  \end{itemize}
\end{thmmain}

\begin{proof}
  Consider the nonlinear iteration
  \begin{equation}\label{eq:iter}
    u^{k+1} = u^k - \mc{L}_{u^0}^{-1} \mc{F}( u^k ),
  \end{equation}
  with initial condition $u^0$ given in \eqref{eq:defu0}.  

  Consider the first iteration ($k = 0$ in \eqref{eq:iter}), we get
  \begin{equation*}
    \norm{u^1 - u^0}_{(H_n^2)^3} = \norm{\mc{L}_{u^0}^{-1} \mc{F}( u^0
      )}_{(H_n^2)^3} \leq \norm{\mc{L}_{u^0}^{-1}}_{\ms{L}((L_n^2)^3,
      (H_n^2)^3)} \norm{\mc{F}( u^0 )}_{(L_n^2)^3}.
  \end{equation*}
  By Proposition~\ref{prop:approximate} and
  Corollary~\ref{coro:invL0}, there exists constant $C_1$ such that
  \begin{equation*}
    \norm{u^1 - u^0}_{(H_n^2)^3} \leq C_1 \veps^3 / 2. 
  \end{equation*}

  Suppose we have proved for all $k \leq k_0$, 
  \begin{equation}\label{eq:Mball}
    \norm{u^k - u^0}_{(H_n^2)^3} \leq C_1 \veps^3,
  \end{equation}
  by the iteration scheme \eqref{eq:iter}, we have for all $k \leq
  k_0$,
  \begin{equation*}
    \begin{aligned}
      u^{k+1} - u^{k} & = u^{k} - u^{k-1} - \mc{L}_{u^0}^{-1}
      ( \mc{F}(u^{k}) - \mc{F}(u^{k-1}) ) \\
      & = - \mc{L}_{u^0}^{-1} \bigl( \mc{F}(u^{k}) - \mc{F}(u^{k-1}) -
      \mc{L}_{u^0} ( u^{k} - u^{k-1}) \bigr).
    \end{aligned}
  \end{equation*}
  Hence, by \eqref{eq:Mball} and Lemma~\ref{lem:diffL},
  \begin{equation*}
    \norm{ u^{k+1} - u^{k} }_{(H_n^2)^3} \lesssim C(C_1) \veps^{3/2} 
    \norm{ u^{k} - u^{k-1} }_{(H_n^2)^3}.
  \end{equation*}
  For $\veps$ sufficiently small so that $ C(C_1)\veps^{3/2} < 1/2$,
  we then have
  \begin{equation*}
    \norm{ u^{k_0+1} - u^0 }_{(H_n^2)^3} \leq \sum_{k = 0}^{k_0}
    \norm{ u^{k+1} - u^{k} }_{(H_n^2)^3} \leq 2 \norm{ u^1 - u^0}_{(H_n^2)^3}
    \leq C_1 \veps^3. 
  \end{equation*}
  Therefore, by induction, we have for all $k\geq 0$, 
  \begin{equation*}
    \norm{ u^{k} - u^0 }_{(H_n^2)^3} \leq C_1 \veps^3,
  \end{equation*}
  and 
  \begin{equation}\label{eq:contract}
    \norm{ u^{k+1} - u^{k}}_{(H_n^2)^3} \leq \frac{1}{2} 
    \norm{ u^{k} - u^{k-1}}_{(H_n^2)^3}.
  \end{equation}
  Because of \eqref{eq:contract}, the iteration
  converges to $u^{\ast}$  and 
  \begin{equation*}
    \mc{F}(u^{\ast}) = 0. 
  \end{equation*}
In addition, we have
  \begin{equation*}
    \norm{ u^{\ast} - u^0 }_{(H_n^2)^3} \leq C_1 \veps^3. 
  \end{equation*}
  The local uniqueness also follows easily from \eqref{eq:contract}. 
\end{proof}

\bibliographystyle{amsplain}

\bibliography{spintfw}

\end{document}